\documentclass[doublecolumn,letterpaper,10pt,conference]{ieeeconf}

\usepackage{graphicx}      
\usepackage{cite}
\let\labelindent\relax 
\usepackage{float}
\usepackage{url}
\usepackage{hyperref}
\usepackage{amsmath}
\usepackage{amsthm}
\usepackage{amssymb}
\usepackage{enumerate}
\usepackage{enumitem}
\usepackage{algorithm} 
\usepackage{algpseudocode} 
\usepackage{cite}
\usepackage[dvipsnames]{xcolor}
\allowdisplaybreaks
\IEEEoverridecommandlockouts
\hypersetup{colorlinks,breaklinks,linkcolor= Blue,urlcolor=Blue,anchorcolor=Blue,citecolor=Blue}

\newcommand{\ex}[1]{\mathbb{E} \left[ #1 \right]}

\newtheorem{definition}{Definition}
\newtheorem{assumption}{Assumption}
\newtheorem{theorem}{Theorem}

\newtheorem{lemma}[theorem]{Lemma}
\newtheorem{remark}{Remark}

\renewcommand{\leq}{\leqslant}
\renewcommand{\geq}{\geqslant}

\title{\LARGE \bf 
Local Stochastic Factored Gradient Descent\\ for Distributed 
Quantum State Tomography}
\author{Junhyung Lyle Kim$^{\star}$, Mohammad Taha Toghani$^{\ddagger}$, C\'{e}sar A. Uribe$^\ddagger$, Anastasios Kyrillidis$^{\star}$
\thanks{$^\star$Department of Computer Science, Rice University, Houston, TX, USA. \{\href{mailto:jlylekim@rice.edu}{jlylekim},~\href{mailto:anastasios@rice.edu}{anastasios}\}@rice.edu}
\thanks{$^\ddagger$Department of Electrical and Computer Engineering, Rice University, Houston, TX, USA. \{\href{mailto:mttoghani@rice.edu}{mttoghani},~\href{mailto:cauribe@rice.edu}{cauribe}\}@rice.edu}
}

\begin{document}
\maketitle

\begin{abstract}
We propose a distributed Quantum State Tomography (QST) protocol, named  Local Stochastic Factored Gradient Descent (Local SFGD), to learn the low-rank factor of a density matrix over a set of local machines. QST is the canonical procedure to characterize the state of a quantum system, which we formulate as a stochastic nonconvex smooth optimization problem. 
Physically, the estimation of a low-rank density matrix helps characterizing the amount of noise introduced by quantum computation. Theoretically, we prove the local convergence of Local SFGD for a general class of restricted strongly convex/smooth loss functions. 
Local SFGD converges locally to a small neighborhood of the global optimum at a linear rate with a constant step size, while it locally converges exactly at a sub-linear rate with diminishing step sizes. With a proper initialization, local convergence results imply global convergence. We validate our theoretical findings with numerical simulations of QST on the Greenberger-Horne-Zeilinger (GHZ) state.
\end{abstract}

\section{Introduction}

A fully-functional fault-tolerant quantum computer faces many technical hurdles. For instance, using superconducting materials technology, quantum computers must remain cooled at a very low temperature---almost absolute zero---to preserve coherence \cite{park2017second}. Moreover, environmental noise from the electronics controlling the quantum system can disrupt the coherence of its qubits.
Thus, the behavior of current quantum computer implementations needs to be characterized, verified, and rigorously certified, before their widespread commercial use~\cite{preskill2018quantum}. 

Quantum State Tomography (QST) is the canonical procedure to characterize the state of a quantum system at various steps of 
a given computation~\cite{altepeter2005photonic}. 
In particular, experimental quantum physicists design quantum circuits that in theory lead to a specific target pure 
state; then, they 
compare the prepared (input) state and the reconstructed (output) state. 
To do so, measurements are taken on independently prepared copies of the state of a quantum system, and then used to estimate the unknown state by post-processing the data~\cite{altepeter2005photonic}.
However, the description complexity of a quantum state grows \emph{exponentially with the number of qubits}, leading to challenging data acquisition, processing, and storage.
Therefore, as the number of qubits and quantum gates increases, so does the need for efficient, robust, and experimentally-accessible protocols to benchmark quantum information processors.

A quantum state can be represented by a density matrix $\rho$ which is a complex, positive semi-definite (PSD) matrix with unit trace. 
The goal of QST is to design protocols that estimate 
$\rho$. For an $n$-qubit mixed state\footnote{Mixed states are the most general way to express a quantum state.} $\Psi \in \mathbb{C}^{2^n}$, 
its density matrix can be written as a mixture of $r$ pure states: {$\rho = \sum_{k}^r p_k \Psi_k \Psi_k^\dagger \in \mathbb{C}^{2^n \times 2^n},$} where $(\cdot)^\dagger$ denotes the complex conjugate operator. 
Here, {$p_k$} is the probability of finding $\rho$ in the pure state {$\Psi_k$}. 
Given these definitions, QST can be formulated as the estimation of a low-rank density matrix $\rho^\star \in \mathbb{C}^{d \times d}$ on an $n$-qubit Hilbert space with dimension $d=2^n$, through the following $\ell_2$-norm optimization problem: 
\begin{equation}\label{eq:obj-qst}
\begin{aligned}
& \min_{\rho \in \mathbb{C}^{d \times d}}
& & F(\rho) := \tfrac{1}{2m} \|\mathcal{A}(\rho) - y\|_2^2 \\
& \text{subject to}
& & \rho \succeq 0, ~\text{rank}(\rho) \leq r,
\end{aligned}
\end{equation}
where $\mathcal{A}: \mathbb{C}^{2^n \times 2^n} \rightarrow \mathbb{R}^m$ is the linear sensing map such that {$\mathcal{A}(\rho)_k = \text{Tr} (A_k \rho )$, for $k = 1, \dots, m$}. The sensing map used in QST has a particular structure: it is the Kronecker product of Pauli matrices {$A_k$},
and is closely related to how quantum computers take measurements in practice~\cite{liu2011universal}. 


The exponential dependency on the number of qubits implies that $\rho$ has more than a trillion entries for a $20$-qubit system. Storing this matrix demands tens of terabytes of memory, which is only available as distributed memory in sizable clusters. 
{
Analogously, a quantum system with more than 30 qubits would demand $100\times$ more memory than the one present in the world’s fastest supercomputer.}



To alleviate these challenges, 
we study the following distributed optimization problem to be jointly solved over a set of $M$ machines:
\vspace{-2mm}
\begin{equation}\label{eq:str-cvx-obj}
\begin{aligned}
& \min_{X \in \mathbb{R}^{d \times d}} 
& & \bigg\{ f(X) = \frac{1}{M} \sum_{i=1}^M f_i(X) \bigg\} \\
& \text{subject to}
& & X \succeq 0, {~\text{rank}(X) \leq r},
\end{aligned}
\end{equation}
In~\eqref{eq:str-cvx-obj}, \mbox{$f_i(X) := \mathbb{E}_{j \sim \mathcal{D}_i} \big[ f_i^j(X) \big]$}, and $f_i^j(X)$ is the loss function evaluated at the $j$-th observation of the locally stored dataset of machine $i$,  which follows the distribution $\mathcal{D}_i.$\footnote{We assume the homogeneous data case where $\mathcal{D}_i = \mathcal{D}$ for all $i$.}
The function $f_i: \mathbb{R}^{d \times d} \rightarrow \mathbb{R}$ is a (restricted) strongly convex/smooth differentiable function, and $X \succeq 0$ is the set of positive semi-definite matrices {with $\text{rank}(X)\leq r$}.\footnote{We provide theory for the real case; extensions to complex domains can be obtained with complex conversions and Wirtinger derivatives \cite{gunning2009analytic}.} 





To solve \eqref{eq:str-cvx-obj}, we introduce the \textit{Local Stochastic Factored Gradient Descent} (Local SFGD) algorithm, and prove its convergence. To the best of our knowledge, this is the first work that studies Local SGD in the non-convex factorized objective, and 
provides convergence in terms of the distance to the optimal model parameter.
Our contributions can be summarized as:
\begin{itemize}[leftmargin=*]
    \item We introduce a distributed problem setup for 
    QST as an instance of~\eqref{eq:str-cvx-obj}.
    \item We propose Local SFGD, a distributed algorithm that uses matrix factorization and utilizes local stochastic gradient steps for the minimization of a non-convex function. 
    \item We provide local convergence guarantees for Local SFGD  for restricted strongly convex/smooth losses, which is of independent interest, and subsumes the QST problem as a special case.
    \item We corroborate our theoretical findings with numerical simulations of QST for the Greenberger-Horne-Zeilinger (GHZ) state.
\end{itemize}
{
This paper is organized as follows. Section~\ref{sec:prelim} reviews the QST protocol, and sets up the non-convex distributed objective. Section~\ref{sec:algo} introduces the Local SFGD algorithm, followed by Section~\ref{sec:main-results} where we provide the main theoretical results along with the proofs. Lastly, in Section~\ref{sec:experiments}, we use Local SFGD for the reconstruction of the GHZ state.
}

\section{Preliminaries}
\label{sec:prelim}

Classically, the sample complexity $m$ for reconstructing $\rho^\star \in \mathbb{C}^{d \times d}$ is $O(d^2),$ where $d$ itself grows exponentially with $n$.
%
To address such large sample complexity requirements, we use low-rankness as prior, as many lab-constructed density matrices have low-rank structure, including the maximally-entangled Greenberger-Horne-Zeilinger (GHZ) state \cite{zhao2021creation}.
While the low-rank constraint is non-convex, it provides a significant reduction in the sample complexity. Under appropriate assumptions,
a rank-$r$ density matrix can be reconstructed with $m = O(r \cdot d \cdot \text{poly} \log(d))$ measurements, 
instead of \mbox{$m = O(d^2)$}~\cite{gross2010quantum}.

{
We propose to solve a factorized version of \eqref{eq:obj-qst} to efficiently handle its low-rank constraint, following \cite{kyrillidis2018provable, kim_2021_fast}:}
\begin{equation}\label{eq:factobj-qst}
\min_{U \in \mathbb{C}^{d \times r}} ~ G(U) := F(UU^\dagger) = \tfrac{1}{2m}  \|\mathcal{A}(UU^\dagger) - y\|_2^2.
\end{equation}
%
{
In 
\eqref{eq:factobj-qst}, we parametrize the low-rank density matrix $\rho$ by its factor $U \in \mathbb{C}^{d \times r}$}.
By rewriting $\rho = UU^\dagger$, 
both the {PSD and the low-rank constraints}
are automatically satisfied, leading to the unconstrained non-convex formulation in \eqref{eq:factobj-qst}. 
Moreover, working in the factored space improves time and space complexities \cite{kyrillidis2018provable, kim_2021_fast, tu2016low}.
%
However, even with the reduced sample complexity $m = O(r \cdot d \cdot \text{poly} \log(d))$, 
linear dependency on $d = 2^n$ makes computation infeasible, e.g., for $n=20$ and rank $r=100$, the reduced sample complexity still reaches $2.02 \times 10^{10}$. 

To handle this explosion of data, we consider the setting where the measurements $y \in \mathbb{R}^m$ and the sensing matrices $\mathcal{A}: \mathbb{C}^{d\times d} \rightarrow \mathbb{R}^m$ from a central quantum computer \emph{are locally stored across $M$ different classical machines}. These classical machines perform some local operations based on their local data, and communicate back and forth with the central quantum server to reconstruct a density matrix.

The \emph{distributed} QST problem can be written as: 
\begin{equation}\label{eq:dist-qst-obj-1}
\begin{split}
    \min_{U \in \mathbb{C}^{d \times r}}  \Big\{ g(U) &= \frac{1}{M} \sum_{i=1}^M g_i(U) \Big\},    \\
  \text{where} \quad g_i(U) &:= \mathbb{E}_{j \sim \mathcal{D}_i} \|\mathcal{A}_i^j (UU^\dagger) - y_i^j \|_2^2, 
\end{split}
\end{equation}
with $j$ being a random variable that follows a distribution $\mathcal{D}_i$ for machine $i$. 
{
In the next section, we introduce our approach to solve \eqref{eq:dist-qst-obj-1}, which can be more generally applied to \eqref{eq:str-cvx-obj}.}

\section{Algorithms}
\label{sec:algo}
We now introduce the 
\emph{Local Stochastic Factored Gradient Descent} (Local SFGD) 
algorithm.
We review the \emph{Factored Gradient Descent} (FGD) algorithm \cite{kyrillidis2018provable, kim_2021_fast, tu2016low}  
and its stochastic variant \cite{zeng_global_2019}, on which {the} Local SFGD is based.

\noindent
\textbf{$\diamond$ Factored Gradient Descent (FGD)}.
A common approach to solve the factorized non-convex objective in \eqref{eq:factobj-qst} in \emph{centralized} settings is to use gradient descent on the factor $U$:
\begin{align}
  U_{t+1} &= U_t - \eta_t \nabla G(U_t) = U_t - \eta_t \nabla F(U_t U_t^\dagger) \cdot U_t \nonumber \\
  &= {U_t - \frac{\eta_t}{m} \Big( \sum_{k=1}^m \left\{ \text{Tr}(A_k U_t U_t^\dagger) - y_k \right\} A_k \Big) \cdot U_t,} \label{eq:fgd}
\end{align}
where $\eta_t>0$ is the step size.
From \eqref{eq:fgd}, we can see that a pass over full data is required to compute the gradient
on every iteration. This can be computationally challenging or even infeasible when $m$ is large, which is almost always the case for QST, even for moderate number of qubits $n$.

\noindent
\textbf{$\diamond$ Stochastic Factored Gradient Descent (SFGD)}.
A simple and effective way to mitigate this burden is to use \emph{Stochastic Factored Gradient Descent} (SFGD), which replaces the true gradient $\nabla G$ with an unbiased estimator $H$.
For instance, one can use the following SFGD update:
\begin{align}
  U_{t+1} 
  &= U_t - \eta_t \cdot H(U_t) \nonumber \\
  &= {U_t - \frac{\eta_t}{b} \Big( \sum_{k=1}^b  \left\{ \text{Tr}(A_k U_t U_t^\dagger) - y_k \right\} A_k \Big) \cdot U_t} ,\label{eq:sfgd}
\end{align}
which simply uses $b$ measurements instead of $m \gg b$ to approximate $\nabla G$, where the hyperparameter $b$ is the batch size. 
In \cite{zeng_global_2019}, the convergence of SFGD was shown for (restricted) strongly convex/smooth functions. 

From \eqref{eq:sfgd}, one can see that SFGD is amenable to parallelization, simply by replacing $H(U_t)$ with an average of stochastic gradients that are computed independently from the local machines.
This scheme is often the state-of-the-art in distributed learning problems \cite{abadi2016tensorflow, goyal2017accurate}.
However, it exhibits a major drawback: on every iteration, each machine has to send the local (stochastic) gradient to the server, and receive back the aggregated model parameter. {Such communication is much more expensive than---typically about 3 orders of magnitude---the local computations that each machine has to perform \cite{lan2020communication}.}

\noindent
\textbf{$\diamond$ Our approach: Local SFGD}.
There are two main approaches to resolve the aforementioned communication overhead.
One is to reduce the number of transmitted bits via gradient compression schemes, such as quantization \cite{alistarh2017qsgd} or sparsification \cite{hoefler2021sparsity}.
The other is to increase the amount of local iterations performed on each machine, in order to reduce the total communication rounds. The latter approach is called Local SGD, and was shown to outperform (parallel) SGD in some settings~\cite{stich_local_2019, haddadpour2019local,woodworth2020local, khaled_tighter_2020}.







In this work, we introduce the Local SFGD 
to estimate the low-rank factor of a density matrix over a set of local machines.
Although our main application is to solve the distributed QST objective in \eqref{eq:dist-qst-obj-1}, {Local SFGD is more generally applicable to the distributed objective in \eqref{eq:str-cvx-obj}}; see Section~\ref{sec:main-results} for details.
Local SFGD is summarized in Algorithm~\ref{alg:local-sfgd}. While there are non-convex results on Local SGD \cite{wang2018cooperative, zhou2017convergence}, they consider a different problem setting,
and only provide convergence in terms of the norm of the gradient.
To the best of our knowledge, this is the first work that studies Local SGD in the non-convex factorized objective, and 
provides convergence in terms of the distance to the optimal model parameters.



%
%
%

Local SFGD produces $M$ sequences in parallel, where $M$ is the number of machines. If a synchronization step happens at time $t$, i.e., $t = t_p$ for some $p \in \mathbb{N},$ then the local parameters at each machine $U_t^i$ are sent to the central server, and their average is computed (line 6). Otherwise, each machine performs (possibly many iterations of) SFGD without communicating with the central server (line 8). 
An important metric to consider for Local SFGD is the maximum time interval between two synchronization time steps: $\max_p |t_p - t_{p+1}|,$ which we assume is bounded by $h \geq 1$; see also Theorems~\ref{thm:lin-conv} and \ref{thm:sub-lin-conv}.
If communication happens on every iteration, i.e.,
$h=1,$
then Algorithm~\ref{alg:local-sfgd} reduces to the (parallel) SFGD in \eqref{eq:sfgd}. 



\begin{algorithm}[tb!]
	\caption{Local SFGD} \label{alg:local-sfgd}
	\begin{algorithmic}[1]
	    \State Set number of iterations $T>0$, synchronization time steps $t_1, t_2, \dots$, and initialize $U_0^i = U_0$ as below: 
	    \vspace{-2mm}
	    \begin{equation}\label{eq:initialization}
            U_0^i = \texttt{SVD} \Big( - \sum_{i=1}^M \tfrac{m_i}{m} \nabla f_i(0) \Big) \quad \forall i \in [M],
        \vspace{-2mm}
        \end{equation}
        where \texttt{SVD} denotes the singular value decomposition.
 		\For {each round $t=0, \ldots T$}
		  \For {in parallel for $i \in [M]$}
		    \State Sample $j_t$ uniformly at random from $[m_i]$.
            \If {$t = t_p$ for some $p \in \mathbb{N}$} 
              \State 
                $U_{t+1}^i = \frac{1}{M} \sum_{i=1}^M \big( U_t^i - \eta_t \nabla g_i^{j_t} (U_t^i) \big)$
            \Else 
             \State 
             $U_{t+1}^i = U_t^i - \eta_t \nabla g_i^{j_t}(U_t^i)$
            \EndIf
          \EndFor
		\EndFor\\
		\Return {$\hat{U}_{T+1} := \frac{1}{M} \sum_{i=1}^M U_{T+1}^i.$}
	\end{algorithmic} 
\end{algorithm}

\section{The Convergence of Local SFGD}
\label{sec:main-results}



We provide local convergence guarantees of Local SFGD in Algorithm~\ref{alg:local-sfgd} for \emph{restricted $\mu$-strongly convex/$L$-smooth objectives}.
Similarly to \eqref{eq:factobj-qst}, as we parametrize $X = UU^\top,$ Problem \eqref{eq:str-cvx-obj} becomes non-convex, where:
\begin{equation} \label{eq:fact-obj}
\min_{U \in \mathbb{R}^{d \times r}}  \Big\{ g(U) = \frac{1}{M} \sum_{i=1}^M g_i(U)  \Big\},
\end{equation}
which now is unconstrained, as {both the PSD and the low-rank constraints are automatically satisfied.}

We assume $f_i$ is a symmetric function: $f_i(X) = f_i(X^\top).$ Then, the gradient of $g_i(U) = f_i(UU^\top)$ simplifies to:\footnote{Without loss of generality, we absorb 2 into $\eta_t$ to use $\nabla g_i(U)$ and $\nabla f_i(UU^\top)U$ interchangeably.}
\begin{small}
\begin{align*}
    \nabla g_i(U) = \left( \nabla f_i(UU^\top) + \nabla f_i(UU^\top)^\top \right)U  = 2 \nabla f_i(UU^\top)U.
\end{align*}
\end{small}
We now state key assumptions used in our main results.
\begin{assumption} 
\label{assu:res-str-cvx-smooth}
    The function $f_i$ is $\mu$-restricted strongly convex and $L$-restricted smooth. That is, $\forall X, Y \succeq 0$ and $\forall i \in [M]$, it holds that
    \begin{subequations}
    \begin{align}
        f_i(Y) \geq f_i(X) &+ \langle \nabla f_i(X), Y{-}X \rangle + \tfrac{\mu}{2} \|X-Y\|_F^2, \tag{I-a} \label{A1:eq:str-cvx} \\
        \text{and}\quad \| \nabla f_i(X) &- \nabla f_i(Y) \|_F \leq L \|X-Y\|_F. \tag{I-b} \label{A1:eq:smooth}
    \end{align}
    \end{subequations}
\end{assumption}


\begin{assumption} \label{assu:stoc-grad}
The stochastic gradient $\nabla g_i^j$ is unbiased, has a bounded variance, and is bounded {in expectation}, $\forall i \in [M].$
That is,
\begin{subequations}
\begin{align} 
    &\mathbb{E}_j \big[ \nabla g_i^j(U) \big] = \nabla g_i(U) \tag{II-a} \label{A2:eq:stoc-grad-unbiased}, \\
    &\mathbb{E}_j \big[ \| \nabla g_i^j(U) - \nabla g_i(U) \|_F^2 \big] \leq \sigma^2  \tag{II-b} \label{A2:eq:stoc-grad-var-bound}, \quad \text{and} \\
    &{\mathbb{E}_j \big[ \| \nabla g_i^j (U) \|_F^2 \big] \leq G^2}, \tag{II-c}  \label{A2:eq:stoc-grad-bounded}
\end{align}
\end{subequations}
where $j$ follows a uniform distribution.
\end{assumption}


Assumptions \eqref{A1:eq:str-cvx} and \eqref{A1:eq:smooth} respectively state that $\mu$-strong convexity and $L$-smoothness hold when we restrict the space of $d \times d$ matrices 
to the set of PSD matrices.
Such assumptions have become standard in optimization analysis, and are significantly weaker than assuming global strong convexity. Importantly, note that we only assume $f_i(X)$ to have such structures---the transformed function $g_i(U)$ in \eqref{eq:fact-obj} typically does not satisfy restricted strong convexity/smoothness \cite{zhou_accelerated_2020}.

Assumptions \eqref{A2:eq:stoc-grad-unbiased} and \eqref{A2:eq:stoc-grad-var-bound} respectively imply that the stochastic gradient is unbiased and has a bounded variance, and both are standard assumptions in stochastic optimization \cite{bottou2018optimization}. 
{
Assumption \eqref{A2:eq:stoc-grad-bounded} states that the stochastic gradient has a bounded norm, in expectation. This assumption may seem strong when the objective is (unconstrained) strongly convex \cite{nguyen2018sgd}; however, note that Assumption \eqref{A1:eq:str-cvx} is restricted to PSD matrices, and the original Problem \eqref{eq:str-cvx-obj} is constrained.
}



Apart from \eqref{eq:fact-obj} being non-convex, another difficulty that arises by the parametrization $X=UU^\top$ is that the solution can become non-unique.\footnote{Consider reconstructing $X^\star = \begin{bmatrix}
    1 & 1  \\
    1 & 1
\end{bmatrix}$. It can be seen that both $U^\star = [1~1]^\top$ and $\tilde{U}^\star = -[1~1]^\top$ satisfy $U^\star U^{\star \top} = \tilde{U}^\star \tilde{U}^{\star \top} = X^\star.$}
We can remove this ambiguity by defining the following rotation invariant distance metric.
\begin{definition}[Eq.~($3.1$) in~\cite{tu2016low}] 
\label{def:procrustes-dist}
For any $U, V \in \mathbb{R}^{d \times r},$ let
$D(U, V) := \min_{R \in \mathcal{O}} \| U - VR \|_F,$
where $\mathcal{O} \subseteq \mathbb{R}^{r \times r}$ is the set of orthonormal matrices such that $R^\top R = \mathbb{I}_{r \times r}.$ 
\end{definition}
%
%

\begin{remark}
Definition~\ref{def:procrustes-dist} regards all $U, V \in \mathbb{R}^{d \times r}$
to be in the same distance such that $D(U, VR) = D(UR, V) = D(U,V).$ Hence, it defines the equivalence classes $\{UR: R^\top R = \mathbb{I}_{r \times r}\}$ and $\{VR: R^\top R = \mathbb{I}_{r \times r}\}$ \cite{zhou_accelerated_2020}.  
\end{remark}


A crucial component for our convergence analysis is the following lemma, which replaces the role of (strong) convexity in classical convergence analysis of gradient descent: 
\begin{lemma}[Lemma 14 in \cite{bhojanapalli2016dropping}] 
\label{lem:descent-lemma}
    Let Assumption~\ref{assu:res-str-cvx-smooth} hold. 
    Assume that $D^2(U_0^i, U^\star) \leq \frac{\sigma_r (X^\star)}{100 \cdot \kappa \cdot \sigma_1 (X^\star)}$, where $\sigma_k(X^\star)$ is the $k$-th singular value of $X^\star$,\footnote{Without loss of generality, singular values are sorted in descending order.} and $\kappa := \frac{L}{\mu}.$ 
    Then, the following inequality holds:
    \begin{equation}
    \begin{split}
        &\big\langle U_t^i - U^\star R^\star, \nabla g_i(U_t^i) \big\rangle \\
        &\quad\quad \geq \tfrac{2\eta_t}{3} \| \nabla g_i(U_t^i) \|_F^2 + \tfrac{3 \mu}{20} \sigma_r(X^\star) \cdot D^2 (U_t^i, U^\star).
    \end{split}
    \end{equation}
\end{lemma}
\begin{remark}
{
The initialization scheme \eqref{eq:initialization} in Algorithm~\ref{alg:local-sfgd} is modified from \cite[Theorem~11]{bhojanapalli2016dropping} to distributed version, and satisfies the initialization condition in Lemma~\ref{lem:descent-lemma} for small enough $\kappa$; for the QST problem, the Pauli sensing matrices $A_k$ satisfy the Restriced Isometry Property (RIP) \cite{liu2011universal, candes2006near}, implying $\kappa \approx \tfrac{1+\delta}{1-\delta}$, where $\delta \in (0, 1)$ is the RIP constant. 
Hence, by using right prior information (e.g., low-rankness), we can apply the compressed sensing results, implying that $X^\star$ is unique and can be recovered exactly \cite{recht2010guaranteed}.
}
\end{remark}

We are now ready to present the main theoretical results. We first show in Theorem~\ref{thm:lin-conv} that the Local SFGD 
converges locally at a linear rate to a small neighborhood of the global optimum with a constant step size. Then, in Theorem~\ref{thm:sub-lin-conv}, we show the exact local convergence by using an appropriately diminishing step size, at the expense of reducing the convergence rate to a sub-linear rate.


\begin{theorem}[Local linear convergence with constant step size]
\label{thm:lin-conv}
Let Assumptions~\ref{assu:res-str-cvx-smooth}, \ref{assu:stoc-grad},
and the initialization condition of Lemma \ref{lem:descent-lemma}
hold. 
Moreover, let $\eta_t = \eta < \tfrac{1}{\alpha}$ for $t \in [0:T]$ and $\max_p |t_p - t_{p+1}| \leq h$. Then, the output of Algorithm~\ref{alg:local-sfgd} has the following property:
{
\begin{equation} \label{eq:lin-conv}
\begin{split}
    \mathbb{E}\big[D^2(\hat{U}_{T+1}, U^\star)\big]
    &\leq 
    \left( 1 - \eta \alpha  \right)^{T+1} D^2(\hat{U}_0, U^\star) \\
    &\quad + \eta \left(\tfrac{(h-1)^2 G^2}{\alpha} +  \tfrac{\sigma^2}{M \alpha} \right),
\end{split}
\end{equation}
}
where $X^\star$ is the optimum of $f$ over the set of PSD matrices such that $\text{rank}(X^\star) = r$, $U^\star$ is such that $X^\star = U^\star U^{\star \top} $, and $\alpha = \tfrac{3 \mu}{10} \sigma_r(X^\star)$ is a global constant.
\end{theorem}
\begin{remark}
{
In \eqref{eq:lin-conv}, the expectation is with respect to the previous iterates, $\{ \hat{U}_t \}_{t=0}^{T}$.
We make a few remarks about Theorem~\ref{thm:lin-conv}. First, notice the last variance term $\tfrac{\sigma^2}{M\alpha}$, which disappears in the noiseless case, is reduced by the number of machines $M$. Second, we assume a single-batch is used in the proof; by using batch size $b>1$, this term can be further divided by $b$. Lastly, by plugging in $h=1$ (i.e., synchronization happens on every iteration), the first variance term 
disappears, exhibiting similar local linear convergence to SFGD \cite{zeng_global_2019}.
}
\end{remark}


\begin{proof}
We start with the following auxiliary lemma. 
\begin{lemma} \label{lem:avg-local-diff}
{
Let Assumptions~\ref{assu:res-str-cvx-smooth} and 
\eqref{A2:eq:stoc-grad-bounded} hold}. Then, the output of Algorithm~\ref{alg:local-sfgd} with 
$\max_p |t_{p+1} - t_p| \leq h$ satisfies:
\begin{align} \label{eq:avg-local-diff}
    {\frac{1}{M} \sum_{i=1}^M \mathbb{E} \big[ \| \hat{U}_t - U_t^i \|_F^2 \big]  \leq (h-1)^2 \eta_t^2 G^2.}
\end{align}
\end{lemma}
The proof is almost identical to \cite[Lemma 3.3]{stich_local_2019}, and hence is omitted.
{Throughout the proof,
we use the notations:}
\begin{align*}
    U_{t+1}^i = U_t^i - \eta_t g_t^i   \quad \text{and} \quad  \hat{U}_{t+1} = \hat{U}_t - \eta_t g_t,
\end{align*}
where $\hat{U_t} = \frac{1}{M} \sum_{i=1}^M U_t^i$, i.e., the average across different machines at time $t$. We denote the stochastic gradient of machine $i$ at time $t$ with $g_t^i := \nabla f_i^{j_t} (U_t^i U_t^{i \top}) U_t^i = \nabla g_i^{j_t}(U_t^i)$, and the average of stochastic gradients across machines with $g_t = \frac{1}{M} \sum_{i=1}^M g_t^i$. Finally, we denote $\ex{g_t} = \bar{g_t}.$

We first decompose the distance of $D^2(\hat{U}_{t+1}, U^\star)$:
\begin{align}
    D^2(\hat{U}_{t+1}, U^\star) &= \min_{R \in \mathcal{O}} \|\hat{U}_{t+1} {-} U^\star R\|_F^2  
    \leq \| \hat{U}_{t+1} {-} U^\star R^\star\|_F^2 \nonumber \\
    &= \|\hat{U}_t - U^\star R^\star - \eta_t \bar{g}_t\|_F^2 + \eta_t^2 \|\bar{g}_t - g_t \|_F^2 \nonumber \\
    &\quad+2\eta_t \langle \hat{U}_t - U^\star R^\star - \eta_t \bar{g}_t, \bar{g}_t - g_t \rangle.
    \label{eq:main-rec}
\end{align}

The first term in \eqref{eq:main-rec} can be further decomposed to:
\begin{align} 
     \|\hat{U}_t - U^\star R^\star \|_F^2  + \eta_t^2 \|\bar{g}_t\|_F^2
    - 2\eta_t \langle \hat{U}_t - U^\star R^\star, \bar{g}_t \rangle. \label{eq:main-rec-first}
\end{align}
We bound the second and the third terms of \eqref{eq:main-rec-first} separately. For the second term, by Jensen's inequality, we have:
\begin{align}
    \|\bar{g}_t\|_F^2 &= \big\| \tfrac{1}{M} \sum_{i=1}^M \nabla g_i(U_t^i) \big\|_F^2
    \leq \tfrac{1}{M} \sum_{i=1}^M \|\nabla g_i(U_t^i)\|_F^2.
    \label{eq:grad-norm-bound}
\end{align}
%
For the third term, we decompose further to have:
\begin{align}
    \big\langle \hat{U}_t {-} U^\star R^\star, \tfrac{1}{M} \sum_{i=1}^M \nabla g_i(U_t^i) \big\rangle &= \tfrac{1}{M} \sum_{i=1}^M \big\langle \hat{U}_t {-} U_t^i , \nabla g_i(U_t^i) \big\rangle  \nonumber\\
    &\hspace{-1.8cm}+ \tfrac{1}{M}\sum_{i=1}^M \big\langle U_t^i - U^\star R^\star, \nabla g_i(U_t^i) \big\rangle.
    \label{eq:descent-term}
\end{align}
We again bound the two terms in \eqref{eq:descent-term} separately. Using $\langle A, B \rangle \geq - \tfrac{\delta}{2} \|A\|_F^2 - \tfrac{1}{2\delta} \|B\|_F^2,$ The first term admits:
\begin{align*}
    &\tfrac{1}{M} \sum_{i=1}^M \big\langle \hat{U}_t - U_t^i , \nabla g_i(U_t^i) \big\rangle \\
    &{\stackrel{}{\geq} \tfrac{1}{M} \sum_{i=1}^M \big(  - \tfrac{\delta}{2} \| \hat{U}_t - U_t^i \|_F^2 - \tfrac{1}{2\delta} \| \nabla g_i(U_t^i) \|_F^2 \big)}
\end{align*}
By Lemma~\ref{lem:descent-lemma}, the second term in \eqref{eq:descent-term} admits:
\begin{align*}
    &\tfrac{1}{M}\sum_{i=1}^M \big\langle U_t^i - U^\star R^\star, \nabla g_i(U_t^i) \big\rangle \\
    & \geq \tfrac{1}{M}\sum_{i=1}^M  \tfrac{2\eta_t}{3} \| \nabla g_i(U_t^i) \|_F^2 + \tfrac{1}{M}\sum_{i=1}^M \tfrac{3\mu\cdot \sigma_r(X^\star)}{20} D^2 (U_t^i, U^\star) \\
    & \geq \tfrac{1}{M}\sum_{i=1}^M  \tfrac{2\eta_t}{3} \| \nabla g_i(U_t^i) \|_F^2 +  \tfrac{3\mu\cdot \sigma_r(X^\star)}{20} D^2 (\hat{U}_t, U^\star), 
\end{align*}
where we used convexity of $D^2(\cdot, \cdot)$ in the second inequality.

Combining above two bounds into \eqref{eq:descent-term}, we have \vspace{-2mm}
\begin{align}
    &\langle \bar{g}_t, \hat{U}_t - U^\star R^\star \rangle  
    \nonumber\\
    & \geq \tfrac{1}{M} \sum_{i=1}^M \big( - \tfrac{\delta}{2} \| \hat{U}_t - U_t^i \|_F^2 - \tfrac{1}{2\delta} \| \nabla g_i(U_t^i) \|_F^2 \big) \nonumber \\
    & + \tfrac{1}{M}\sum_{i=1}^M  \tfrac{2\eta_t}{3} \| \nabla g_i(U_t^i) \|_F^2  +  \tfrac{3\mu\cdot \sigma_r(X^\star)}{20}  D^2 (\hat{U}_t, U^\star).
     \label{eq:descent-term-bound}
\end{align}
%
Substituting \eqref{eq:grad-norm-bound} and \eqref{eq:descent-term-bound} into \eqref{eq:main-rec-first}, we have
\vspace{-2mm}
\begin{align}
    &\|\hat{U}_t - U^\star R^\star - \eta_t \bar{g}_t\|_F^2 \nonumber \\
    &\quad\leq \left( 1 - \eta_t \cdot \tfrac{3 \mu}{10} \sigma_r(X^\star) \right) \|\hat{U}_t - U^\star R^\star \|_F^2 \nonumber \\
    &\quad\quad+ \tfrac{1}{M} \sum_{i=1}^M \big[ \big( \tfrac{\eta_t}{\delta} - \tfrac{\eta_t^2}{3} \big) \| \nabla g_i(U_t^i) \|_F^2 + \eta_t \delta \| \hat{U}_t - U_t^i\|_F^2 \big] \nonumber \\
    &\hspace{1mm}{\stackrel{\delta = 4/\eta_t}{=} \left( 1 - \eta_t \alpha\right) \|\hat{U}_t {-} U^\star R^\star \|_F^2} \nonumber \\ 
    &\quad\quad {+ \tfrac{1}{M} \sum_{i=1}^M \big[ \eta_t \big( \tfrac{\eta_t}{4} - \tfrac{\eta_t}{3} \big) \| \nabla g_i(U_t^i) \|_F^2 + 4 \| \hat{U}_t - U_t^i\|_F^2 \big] } \nonumber \\
    &\quad{\stackrel{}{\leq} \left( 1 - \eta_t \alpha\right) \|\hat{U}_t {-} U^\star R^\star \|_F^2 + \tfrac{4}{M} \sum_{i=1}^M \| \hat{U}_t - U_t^i\|_F^2,}
     \label{eq:main-rec-first-term-bound}
\end{align}
{
where in the equality we defined $\alpha := \tfrac{3 \mu}{10} \sigma_r(X^\star),$ and in the last inequality we used that $\tfrac{\eta_t}{4} - \tfrac{\eta_t}{3} < 0$.} 

Substituting \eqref{eq:main-rec-first-term-bound} into \eqref{eq:main-rec} and taking expectations conditional on previous iterates, 
and using $\ex{g_t} = \bar{g}_t$, 
we get
\begin{align}
    &\mathbb{E} [D^2(\hat{U}_{t+1}, U^\star)]
    = \|\hat{U}_t {-} U^\star R^\star {-} \eta_t \bar{g}_t\|_F^2 + \eta_t^2 \mathbb{E}[\|\bar{g}_t {-} g_t \|_F^2] \nonumber \\
    &{\hspace{0cm}\stackrel{\eqref{eq:main-rec-first-term-bound}}{\leq} \left( 1 - \eta_t \alpha \right) \|\hat{U}_t - U^\star R^\star \|_F^2} \nonumber \\
    &{\quad\quad+ \tfrac{4}{M} \sum_{i=1}^M \mathbb{E} \big[ \| \hat{U}_t - U_t^i \|_F^2 \big] + \eta_t^2 \mathbb{E}[\|\bar{g}_t - g_t \|_F^2] } \nonumber \\
    &{\hspace{0cm}\stackrel{\eqref{eq:avg-local-diff}}{\leq} \left( 1 - \eta_t \alpha \right) \|\hat{U}_t - U^\star R^\star \|_F^2 } \nonumber \\ 
    &{\quad\quad+ 4 \eta_t^2 (h-1)^2 G^2 + \eta_t^2 \mathbb{E}[\|\bar{g}_t - g_t \|_F^2].} \nonumber
\end{align}
where 
the last inequality is by Lemma~\ref{lem:avg-local-diff}.
%
We further have:
\vspace{-3mm}
\begin{align*}
    \mathbb{E}[\|\bar{g}_t - g_t \|_F^2] &= \mathbb{E} \Big[ \Big\| \tfrac{1}{M}\sum_{i=1}^M \big( \nabla g_i^{j_t} (U_t^i) - \nabla g_i (U_t^i) \big) \Big\|_F^2 \Big]  \\
    &\hspace{-1.7cm}\stackrel{}{\leq} \tfrac{1}{M^2}\sum_{i=1}^M \mathbb{E}[\| \nabla g_i^{j_t} (U_t^i) - \nabla g_i(U_t^i) \|_F^2] \stackrel{\eqref{A2:eq:stoc-grad-var-bound}}{\leq} \tfrac{\sigma^2}{M},
\end{align*}
where we used $\text{Var} (\sum_{m=1}^M X_m) = \sum_{m=1}^M \text{Var}(X_m)$ for independent random variables. 

We now arrive at the iteration invariant bound:
{
\begin{align}
    &\mathbb{E}[D^2(\hat{U}_{t+1}, U^\star)] \nonumber \\
    &\hspace{0cm}\leq \left( 1 - \eta_t \alpha  \right) D^2(\hat{U}_t, U^\star) + \eta_t^2 \big( 4(h-1)^2 G^2  + \tfrac{\sigma^2}{M} \big). \label{eq:iteration-invariant}
\end{align}
}
%
Lastly, we unfold \eqref{eq:iteration-invariant} for $T$ iterations, and using \newline $\sum_{t=0}^T (1-\eta_t \alpha)^t \leq \sum_{t=0}^\infty (1-\eta_t \alpha)^t = \frac{1}{\eta_t \alpha},$ we obtain
{
\begin{align}
    \mathbb{E}[D^2(\hat{U}_{T+1}, U^\star)] &\leq 
    \left( 1 - \eta_t \alpha  \right)^{T+1} D^2(\hat{U}_0, U^\star) \nonumber \\
    &\quad+ \eta_t \left( \tfrac{4(h-1)^2 G^2}{\alpha} +  \tfrac{\sigma^2}{M \alpha} \right), \nonumber
\end{align}
}
which completes the proof.
\end{proof}



\begin{theorem}[Local sub-linear convergence with diminishing step size]
\label{thm:sub-lin-conv}
Let Assumptions~\ref{assu:res-str-cvx-smooth}, \ref{assu:stoc-grad},
and the initialization condition of Lemma \ref{lem:descent-lemma}
hold. Moreover, let $\eta_t = \frac{2}{\alpha(t+2)}$ for $t \in [0:T]$ and $\max_p |t_p - t_{p+1}| \leq h$. Then, the output of Algorithm~\ref{alg:local-sfgd} has the following property:
\begin{equation}
\begin{split}
    \mathbb{E}\big[D^2(\hat{U}_{T+1}, U^\star)\big]
    &\leq \tfrac{4C}{\alpha (T+3)},
\end{split}
\end{equation}
where $X^\star$ is the optimum such that $\text{rank}(X^\star) = r$, $U^\star$ is such that $X^\star = U^\star U^{\star \top}$, and $\alpha = \tfrac{3 \mu}{10} \sigma_r(X^\star)$ and {$C = 4(h-1)^2G^2  + \tfrac{\sigma^2}{M}$} are global constants.
\end{theorem}

\begin{proof}

We claim the following, and prove by induction:
\begin{align}\label{eq:induction-hypothesis}
    D^2 (\hat{U}_t, U^\star) \leq \tfrac{4C}{\alpha^2 (t+2)}, \quad \text{with} \quad \eta_t = \tfrac{2}{\alpha(t+2)}.
\end{align}
We start from the iteration invariant bound in \eqref{eq:iteration-invariant}:
\begin{align*}
    \mathbb{E}\big[D^2(\hat{U}_{t+1}, U^\star)\big]
    &\leq  \left( 1 - \eta_t \alpha  \right) D^2(\hat{U}_t, U^\star) + \eta_t^2 \cdot C.
\end{align*}
%
For the base case $t=0$, we have 
\begin{align*}
   & \mathbb{E}\big[D^2(\hat{U}_1, U^\star)\big] \leq (1-\eta_0 \alpha) D^2 (\hat{U}_0, U^\star) + \eta_0^2 \cdot C \\
& = \left( 1 - \tfrac{1}{\alpha} \cdot \alpha \right) D^2 (\hat{U}_0, U^\star) + \tfrac{C}{\alpha^2}= \tfrac{C}{\alpha^2} \leq \tfrac{4C}{3\alpha^2}.
\end{align*}
Now, we proceed to the inductive step. Assuming \eqref{eq:induction-hypothesis} holds for time step $t$, we want to prove the same holds for the time step $t+1$.
Starting from \eqref{eq:iteration-invariant} again, we have:
\begin{align*}
    &\mathbb{E}\big[D^2(\hat{U}_{t+1}, U^\star)\big]
    \leq  \left( 1 - \eta_t \alpha  \right) D^2(\hat{U}_t, U^\star) + \eta_t^2 \cdot C \\
    &\stackrel{\eqref{eq:induction-hypothesis}}{\leq} \left( 1 - \tfrac{2}{t+2} \right) \cdot \tfrac{4C}{\alpha^2 (t+2)} + \tfrac{4C}{\alpha^2 (t+2)^2} 
    = 4C \cdot \tfrac{t+1}{\alpha^2 (t+2)} \cdot \tfrac{1}{t+2} \\
    &\leq 4C \cdot \tfrac{t+2}{\alpha^2 (t+3)} \cdot \tfrac{1}{t+2} = \tfrac{4C}{\alpha^2(t+3)},
\end{align*}
where 
in the last inequality we used the fact that $\frac{t+1}{t+2} \leq \frac{t+2}{t+3}.$ This completes the proof.
\end{proof}

\section{Numerical Results}
\label{sec:experiments}


We use Local SFGD to reconstruct the Greenberger-Horne-Zeilinger (GHZ) state, using simulated measurement data from Qiskit.
GHZ state is known as \emph{maximally entangled} quantum state \cite{zhao2021creation}, meaning it exhibits the maximal inter-particle correlation, which does not exist in classical mechanics. 
%
%
We are interested in: $(i)$ how the number of local steps affect the accuracy defined as $\varepsilon = \|\hat{U}_t \hat{U}_t^\top - \rho^\star_{\text{ghz}}\|_F^2$, where $\rho^\star_{\text{ghz}} = U^\star U^{\star \dagger}$ is the true density matrix for the GHZ state;
and $(ii)$ the scalability of the distributed setup for various number of classical machines $M$. 

In Fig.~\ref{fig:local-sfgd}~(Top), we first fix the number of machines $M = 10$ and the number of total synchronization steps to be $100$, and vary the number of local iterations between two synchronization steps, i.e., $h \in \{1, 10, 25, 50, 100, 200\}.$ {We use constant step size $\eta=1$ for all $h$.}
Increasing $h$, i.e., each distributed machine performing more local iterations, leads to faster convergence in terms of the synchronization steps. 
{Notably, the speed up gets marginal: e.g., there is not much difference between $h=100$ and $h=200$, indicating there is an ``optimal'' $h$ that leads to the biggest reduction in the number of synchronization steps. Further, one can notice that higher $h$ leads to slightly worse final accuracy---this is consistent with \eqref{eq:lin-conv} in Theorem~\ref{thm:lin-conv}, where the first variance term that depends on $G^2$ disappears with $h=1$. Finally, note that $\varepsilon$ does not decrease below certain level due to the inherent finite sampling error of quantum measurements~\cite{crawford2021efficient}.}

In Fig.~\ref{fig:local-sfgd}~{(Bottom)}, we plot the number of synchronization steps to reach $\varepsilon \leq 0.05,$ while fixing 
$h=20$. We vary the number of workers $M \in \{5, 10, 15, 20\}$, where each machine gets $200$ measurements. There is a significant speed up from $M=5$ to $M=15$, while for $M=20,$ it took one more syncrhonization step compared to $M=15,$ which is likely due to the stochasticity of SFGD within each machine.



\begin{figure}[tb!] 
    \centering
    \includegraphics[width=0.85\linewidth]{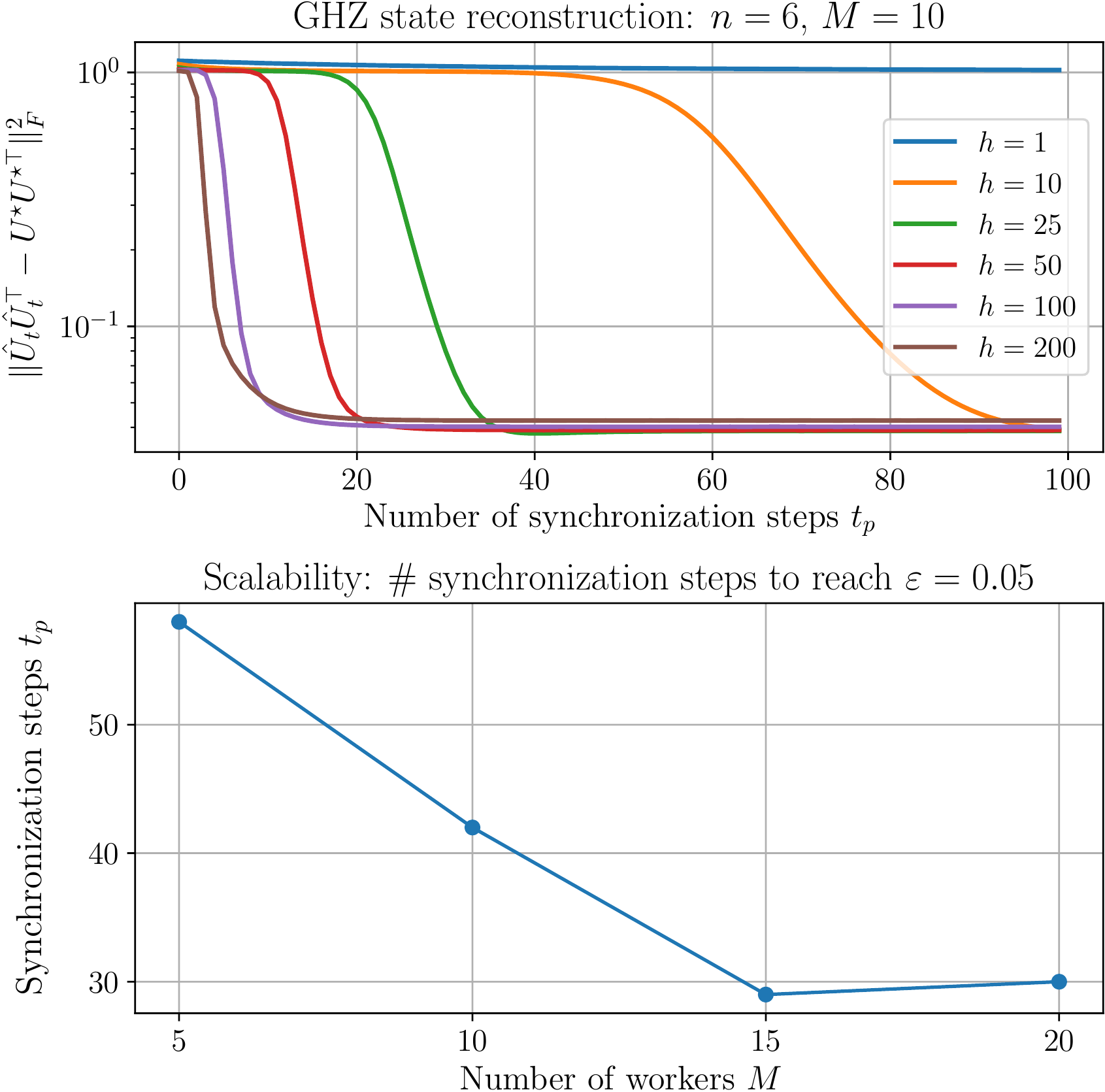}
    \caption{Top: Convergence speed as a function of number of synchronization steps $t_p$ for various number of local iterations. {Bottom}: number of synchronization steps to reach $\varepsilon \leq 0.05$ as a function of number of workers $M$. The batch size $b = 50$ is used for all cases.}
    \label{fig:local-sfgd}
\end{figure}

\section{Conclusion and Future Work}
In this work, we introduced a distributed problem set up for QST as an instance of a general distributed  {optimization problem with PSD/low-rank constraints}. We proposed the Local SFGD, a distributed non-convex algorithm that utilizes local 
steps at each distributed worker to estimate the low-rank factor of a density matrix. We proved the local convergence of Local SFGD for restricted strongly convex/smooth objectives, which can be of independent interest. For future work, extension to the heterogeneous data case as well as the decentralized case with various topologies can be investigated.  






\bibliographystyle{IEEEbib}
\bibliography{references}

\end{document}